\documentclass[pra,twocolumn,10pt,aps]{revtex4-1}
\usepackage[utf8]{inputenc}  
\usepackage[T1]{fontenc}     
\usepackage[british]{babel}  
\usepackage[scaled=0.86]{berasans}  
\usepackage[colorlinks=true,allcolors=blue]{hyperref}  
\usepackage{graphicx} 
\usepackage[babel]{microtype}  
\usepackage{amsmath,amssymb,amsthm,bm,amsfonts,mathrsfs,bbm} 
\usepackage{setspace}
\usepackage{braket}
\usepackage{csquotes}

\usepackage{xspace}  
\usepackage{pgfplots}
\usepackage{verbatim}

\newcommand\id{\leavevmode\hbox{\small1\kern-3.3pt\normalsize1}}

\newtheorem{theorem}{Theorem}

\newtheorem{lemma}[theorem]{Lemma}

\begin{document}


\title{Non-negativity of conditional von Neumann entropy and global unitary operations}

\author{Subhasree Patro}
\email{subhasree.patro@research.iiit.ac.in}
\affiliation{Center for Security, Theory and Algorithmic Research, International Institute of Information 
Technology-Hyderabad, Gachibowli, Telangana-500032, India.}

\author{Indranil Chakrabarty}
\email{indranil.chakrabarty@iiit.ac.in}
\affiliation{Center for Security, Theory and Algorithmic Research, International Institute of Information 
Technology-Hyderabad, Gachibowli, Telangana-500032, India.}

\author{Nirman Ganguly}
\affiliation{$^2$Physics and Applied Mathematics Unit, Indian Statistical Institute, 203 B. T. Road, Kolkata 700108, India.}
\email{nirmanganguly@gmail.com}
\thanks{ At present on leave from Department of Mathematics, Heritage Institute of Technology, Kolkata-107,India}

\date{\today}
\begin{abstract}
Conditional von Neumann entropy is an intriguing concept in quantum information theory. In the present work, we examine the effect of global unitary operations on the conditional entropy of the system. We start with the set containing states with non-negative conditional entropy and find that some states preserve the non-negativity under unitary operations on the composite system. We call this class of states as Absolute Conditional von Neumann entropy Non Negative class (\textbf{ACVENN}). We are able to characterize such states for $2\otimes 2$ dimensional systems.
On a different perspective the characterization accentuates the detection of states whose conditional entropy becomes negative after the global unitary action. Interestingly, we show that this \textbf{ACVENN} class of states forms a set which is convex and compact. This feature enables the existence of hermitian witness operators. With these we can distinguish the unknown states which will have negative conditional entropy after the global unitary operation. We also show that this has immediate application in super dense coding and state merging as negativity of conditional entropy plays a key role in both these information processing tasks. Some illustrations followed by analysis are also provided to probe the connection of such states with absolutely separable (\textbf{AS}) states and absolutely local (\textbf{AL}) states.
\end{abstract}

\pacs{Valid PACS appear here}
\maketitle


\section{\label{sec:level1}Introduction}

Entanglement\cite{einstein} which lies at the heart of quantum mechanics is not only of deep philosophical interest \cite{bell} but also established as the most pivotal resource in various information processing tasks, like teleportation \cite{ben3}, super dense coding \cite{ben2}, key generation \cite{ekert, ben1}, secret sharing \cite{secretSharing}, remote entanglement distribution \cite{sazim-tele} and many more \cite{entanglement and others, bose,broad}.  However , not all entangled states can be directly used for an information processing task, pertinent mentions in this regard are the bound entangled states \cite{bound}. However, these entangled states are available when we go beyond $2 \otimes 2$ and $2 \otimes 3$ system, where we do not have necessary sufficient condition like Peres-Horodecki criterion \cite{peres-hor} for detection of entanglement. Some entangled states have to be processed by local filtering \cite{filt}  before they can be used in a task. As a consequence , telportation witnesses, thermodynamical witnesses\cite{wit} have been devised which can identify useful entangled states for various tasks. In multi qubit systems concepts like 'task oriented entangled' states \cite{pankaj} have been introduced.\\
\indent The ubiquitous role of entanglement in information processing tasks has motivated recent research in the generation of entangled states from separable states. Global unitary operations can play a significant role in this scenario as local unitaries cannot generate entanglement. However, there are some separable states termed as absolutely separable\cite{apss1} from which no entanglement can be produced even with any arbitrary global unitary operation. Characterization of such states has been an active line of research in recent times \cite{apss2}. This notion of ``absoluteness" was extended to define absolutely Bell-CHSH local states and absolute unsteerability \cite{nir1,nir2}. The notion of ``absoluteness" indicates that the state preserves a certain characteristic trait under global unitary transformations. For absolutely separable states it is separability, for absolutely Bell-CHSH local states it is their nature of being Bell-CHSH local.\\
\indent Conditional von Neumann entropy is another such characteristic trait of quantum states. Unlike its classical counterpart this quantity can be negative \cite{Neilsen} , providing yet again a departure from classical information theory. An operational interpretation of the quantum conditional entropy was provided in \cite{winter} , in terms of state merging. The negativity of the conditional entropy also indicates the signature of entanglement, although the converse of the statement is not true as there are entangled states with non-negative conditional entropy. Conditional entropy also plays a key role in dense coding \cite{distdc} , as a bipartite quantum state is useful for dense coding in a sense that it will have quantum advantage if and only if it has a negative conditional entropy.\\
\indent Negativity of conditional entropy being such an important yardstick, our present work probes whether it is always possible to start with a state having non-negative conditional entropy and arrive at a state having negative conditional entropy via global unitaries. We find that there is a class of states which preserve the non-negativity of the conditional entropy under global unitary transformations. 
The characterization also enables one to identify useful states whose conditional entropy becomes negative with a global unitary. It is interesting to find that this class of state \textbf{ACVENN} which preserve the non negativity of the conditional entropy, is convex and compact set. This in principle guarantees us to create the witness operator to detect these states which can arrive at negative conditional entropy in spite of starting with non negative conditional entropy using global unitary operations.
Since separability and non-locality are also important distinctive features of quantum mechanics, we also discuss the connections of these states preserving the non-negativity of conditional entropy under global unitary with the absolutely separable states \textbf{AS} and the recently introduced absolutely Bell-CHSH local states \textbf{AL} \cite{nir1}.\\
\indent Our work has immediate  applications in the information processing tasks like super dense coding \cite{ben2} and state merging \cite{winter}. In super dense coding states with negative conditional entropy gives us quantum advantages while in state merging same states are useful as potential future resource. One starts with some seemingly useless states having non-negative conditional entropy, then using global unitaries as a resource one can turn those states into states having a negative conditional entropy. Since  \textbf{ACVENN} class is convex and compact, it is in principle possible to create witness  operator to detect these transformed states.\\
\indent In section \ref{sec:UD} we give an introduction to all the related concepts that are relevant to this article. In section \ref{sec:CHARAC} we give a general necessary and sufficient condition to characterize \textbf{ACVENN} class of states in the state space of two qubit systems. In section \ref{sec:CONVEXandCOMP}, we show that this  \textbf{ACVENN} class of states is convex and compact which in principle allows to construct the witness operator for identifying those states which do not belong to this class.
In section \ref{sec:REL} we connect this \textbf{ACVENN} class of states with absolutely separable \textbf{AS} and absolutely local state \textbf{AL} states. In section \ref{sec:APPL} we show the potential application of charecterizing such states in various information processing tasks like super dense coding and state merging. Finally we conclude in section \ref{sec:CONCL}. 

\section{\label{sec:UD} Useful definitions and related concepts}
\noindent In this section, we will briefly introduce the various concepts which are going to be useful and are related to the main theme of our paper. We present these concepts in different subsections.\\

\subsection{\label{sec:UD_RHO} General 2-qubit states}
In this work we have considered bloch representation of generalized two qubit states.
A general two qubit state is represented in the canonical form as,

\begin{equation}
\label{generalRho}
\rho=\frac{1}{4}[\mathbb{I}_2\otimes \mathbb{I}_2 + \sum_{i=1}^{3}r_{i}\sigma_{i}\otimes \mathbb{I}_2 + \sum_{i=1}^{3}s_{i}\mathbb{I}_2\otimes\sigma_{i} \\+ \sum_{i,j=1}^{3}t_{ij}\sigma_{i}\otimes\sigma_{j}],
\end{equation}

where $r_i=Tr[\rho(\sigma_{i}\otimes \mathbb{I}_2)]$,  $s_i=Tr[\rho(\mathbb{I}_2\otimes\sigma_{i})]$ are local Bloch vectors. The correlation matrix is given by $T=[t_{ij}]$ where $t_{ij}=Tr[\rho(\sigma_i\otimes\sigma_{j})]$ with [$\sigma_i;\:i$ = $\{1,2,3\}$] are $2\otimes 2$ Pauli matrices and $\mathbb{I}_2$ denotes identity.
\newline
In this paper, we use the notation \textbf{Q} to denote the set of all two-qubit states.\\

\subsection{\label{sec:UD_AS} Separable and Absolutely separable class of states}
When we go beyond  the one qubit system to two qubit system, we come across the notion of entanglement, the states which can not be written as convex combination of tensor product of one qubit systems. The exact complement of this are those states for which composite system can be written as convex combination of tensor product of subsystems. 
However, the definition is not so straightforward when we go beyond two qubit pure states. For a mixed quantum system consisting of two subsystems the general definition of being  separable is if its density matrix can be written as $\sigma_{sep}$=$\sum \lambda_{i} \sigma^{A}\otimes\sigma^{B}$, ($ \sum \lambda_{i}=1, \lambda_{i} \ge 0 $ ), where $\sigma^{A}$ and $\sigma^{B}$ are density matrices for the two subsystems \textit{A} and \textit{B} \cite{peres-hor}.
The set \textbf{S} will denote the class of separable states. Lately people have identified the class of absolutely separable states denoted by \textbf{AS} \cite{apss1,apss2} which are states that remain separable under all global unitary operations, i.e., \textbf{AS} = \{ $\sigma_{as}$ : \textit{U$\sigma_{as}U^\dagger$} is separable $\forall$ $U$\}

\subsection{\label{sec:UD_AL} Local and Absolutely local class of states}
We denote  the set of all states which do not violate the Bell-CHSH inequality by \textbf{L} \cite{clauser}. Recall that any density matrix in two qubits can be written in the canonical form, where $T$ denotes the correlation matrix corresponding to $\rho$. The function $M$($\rho$) is defined as the sum of the maximum two eigenvalues of $T^tT$. Any state with $M$($\rho$) $\leq$ 1 is considered local with respect to the Bell-CHSH inequality \cite{horobell}. Set of states that do not violate Bell-CHSH inequality is denoted by \textbf{L} = \{ $\sigma_L$ : $M$($\sigma_L$) $\leq$ 1 \}. Recently researchers were able to characterize the states which do not violate Bell-CHSH inequality under any global unitary. This set containing these states are denoted by \textbf{AL} \cite{nir1} and is defined by  \textbf{AL}= \{$\sigma_{al}$ : $M$($U\sigma_{al} U^\dagger$) $\leq$ 1 $\forall U$\}.

\subsection{\label{sec:UD_WIT} Witness operator and Geometric form of Hahn-Banach theorem}
A geometric form of the Hahn-Banach theorem states that given a set that is convex and compact, there exists a hyperplane that can separate any point lying outside the set from the given set \cite{Hahn}.
A witness operator $W$  pertaining to a convex and compact set $S$ will be a hermitian operator that satisfies the following conditions: ({1}) $Tr(W\sigma) \geq$ 0, for all states $\sigma$ $\in$ $S$, ({2}) $Tr(W\chi) <$ 0, for any state $\chi$ $\notin$ $S$ \cite{wit}.

\subsection{\label{sec:UD_SDC} Dense coding capacity}
Quantum super dense coding involves in sending of classical information from one sender to the receiver when they are sharing a quantum resource in the form of an entangled state. More specifically, superdense coding is a technique used in quantum information theory to transmit classical information by sending quantum systems. It is quite well known that if we have a maximally entangled state in $H_d\otimes H_d$ as our resource, then we can send $2 \log d$ bits of classical information. In the asymptotic case, we know one can send $\log d + S(\rho)$ amount of bits. It had been seen that the number of classical bits one
can transmit using a non-maximally entangled state in $H_d\otimes H_d$ as a resource is $(1 + p_0\frac{d}{d-1}) \log d$,
where $ p_0 $ is the smallest Schmidt coefficient. However, when the state is maximally entangled in its subspace then one can send up to
$ 2 \log(d-1) $ bits \cite{ben2,distdc}.\\


\subsection{\label{sec:UD_SM}State merging}
Another important information processing task is state
merging. In the classical setting, the idea of state merging
is essentially the following: Consider two parties Alice and
Bob, where Bob has some prior information $B$ and Alice has
some missing information $A$ (where $A$ and $B$ are random
variables). At this point one important question is: If Bob wants to learn about $A$, how
much additional information Alice does need to send
him? It has been shown that only $H (A | B )$ bits suffices.
In the quantum setting, Alice and Bob each possess a
system in some unknown quantum state with joint density
operator $\rho_{AB}$. Assuming that Bob is correlated with Alice,
one asks how much additional quantum information Alice
needs to send him, so that he has knowledge about the
entire state. The amount of partial quantum information \cite{winter}
that Alice needs to send Bob is given by the quantum
conditional entropy, $ S(A|B)= S(\rho_{AB}) - S(\rho_A) $. Ideally  this conditional entropy can be positive ($S(A|B)>0$), negative ($S(A|B)<0$) and zero ($S(A|B)=0$). If it is positive, it means that sender needs to communicate that number of quantum bits to the receiver; if zero it tells there is no need of such communication. However, if it is negative, the sender and receiver gain the same amount of potential for future quantum communication. 

\section{\label{sec:CHARAC} Characterization of Absolute Conditional von Neumann Entropy Non negative (\textbf{ACVENN}) class}

\noindent In this section we will introduce the class of states for which the conditional von Neumann entropy remains non negative even after the application of global unitary operator. The characterization of these states enables us to identify states which can be made useful for some information processing task. The von Neumann entropy of a system $\rho_{AB}$ with two subsystems \textit{A} and \textit{B} is denoted by $S(\rho_{AB})$. The conditional von Neumann entropy for $ \rho_{AB} $ entropy is defined as $ S(\rho_{AB}) - S(\rho_A) $ , where $S(\rho_A)$ denotes the von Neumann entropy of the subsystem \textit{A}. We note the class of states for which the  conditional von Neumann entropy is non negative. We denote this class by \textbf{CVENN} defined by \textbf{CVENN}= \{$\sigma_{cv}$ : S($ \sigma_{cv} $) - S(($ \sigma_{cv})_A) \geq$ 0\}.\\

\noindent \textbf{ACVENN}:The set of states whose conditional von Neumann entropy remains non-negative under any global unitary operations is denoted by \textbf{ACVENN}= \{$ \sigma_{ac}$ : S($U \sigma_{ac} U^\dagger $) - S[($U \sigma_{ac} U^\dagger)_{A}] \geq$ 0, $\forall U$\}.
The von Neumann entropy remains invariant under global unitary transformations, however the conditional entropy can change. We are interested in characterizing the set of states that preserves the non-negativity of the conditional entropy under unitary action on the composite system.

\begin{theorem}
 A state $\sigma_{ac} \in $ \textbf{ACVENN} iff $ S(\sigma_{ac}) \ge 1 $.
\end{theorem}

\begin{proof}
Let $\sigma_{ac} \in $ \textbf{ACVENN}. Then, S($U \sigma_{ac} U^\dagger $) - S[($U \sigma_{ac} U^\dagger)_{A}] \geq$ 0, $\forall U$. This implies, S($\sigma_{ac}$) - S[($U \sigma_{ac} U^\dagger)_{A}] \geq$ 0, $\forall U$, as von Neumann entropy is invariant under changes in the basis of $\sigma_{ac}$, i.e., S($\sigma_{ac}$)= S($U \sigma_{ac} U^\dagger $) with U being any unitary transformation. Hence, we have S($\sigma_{ac}$) $\ge$ S[($U \sigma_{ac} U^\dagger)_{A}]$, $\forall U$. The maximum value of S[$(U \sigma_{ac} U^\dagger)_{A}]$ is obtained at $(U \sigma_{ac} U^\dagger)_{A}$= $\frac{\mathbb{I}}{2}$ and the maximum value is 1. There always exists a unitary that converts the $\sigma_{ac}$ to a Bell diagonal $\sigma_{bell}$ for a given spectrum.
And we know that for a Bell diagonal state the reduced subsystem $(\sigma_{bell})_{A}$ is $\frac{\mathbb{I}}{2}$. Therefore, S($\sigma_{ac}$) $\ge$ S[($U \sigma_{ac} U^\dagger)_{A}]$, $\forall U$ $\Rightarrow$ S($\sigma_{ac}$) $\ge$ S[$(\sigma_{bell})_{A}$]= $S(\frac{\mathbb{I}}{2})$= 1. 
\newline
Conversely let $ S(\sigma_{ac}) \ge 1 $, one can note that the maximum achievable von Neumann entropy of a subsystem is 1 in case of two qubit system. As under a unitary transformation, the entropy of the subsystem alone changes. Hence, for any state $\sigma_{ac}$ whose von Neumann entropy is greater than equal to 1, we know that this state cannot have negative conditional entropy under any global unitary operations. Therefore, any state $\sigma_{ac}$, whose $ S(\sigma_{ac}) \ge 1 $ will $\in $ \textbf{ACVENN}.
\end{proof}

One may quickly note the following observations,
\begin{itemize}
\item Any pure separable state has a non-negative conditional entropy and can be brought by some unitary to a maximally entangled state which now possesses a negative conditional entropy and thus pure separable states can never belong to our desired class. Pure entangled states itself have a negative conditional entropy. Therefore, pure states are not eligible members of \textbf{ACVENN}.

\item The fact that some mixed states will be members of \textbf{ACVENN} is exemplified by the maximally mixed state which remains invariant under any global unitary operation and thus preserves the non-negativity of the conditional entropy. However, the maximally mixed state only constitutes a trivial example and we find that the class contains some very non-trivial states.
\\\\
\end{itemize}

\begin{center}\noindent\textbf{Example : A. Werner State}\end{center}

As an example, we first consider the example of Werner state. The density matrix representation of an Werner state is given by,
\begin{equation}
\sigma_{wer}=(1-p)(\mathbb{I}/4)+ p|\psi\rangle \langle \psi|,
\end{equation}
where, $|\psi\rangle=1/\sqrt{2}(|00\rangle +|11\rangle )$ is the Bell state and $p$ is the classical mixing parameter and $\mathbb{I}$ denotes identity.\\

\begin{figure}[h]
    \centering
    \includegraphics[width=70mm,scale=0.5]{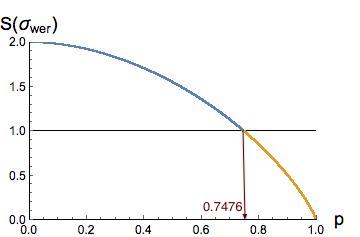}
    \caption{The von Neumann entropy of Werner state $\sigma_{wer}$ against the classical mixing parameter $p$}
    \label{fig:Werner_label}
\end{figure}

In the figure \ref{fig:Werner_label} we have plotted the von Neumann entropy of the Werner state with respect to the mixing parameter $p$. Interestingly, we find that for all values of $p$ $\in$ [0,$\approx$0.7476], we  have $S(\sigma_{wer}) \geq 1$. This clearly indicates the Werner state  for values of $p$ $\in$ [0,$\approx$0.7476] falls within the \textbf{ACVENN} class.\\

\begin{center}\noindent\textbf{Example : B. Bell Diagonal States}\end{center}

Bell-diagonal states can be expressed as, $\sigma_{bell} = \{\vec{0}, \vec{0}, T^b\}$, where $\vec{0}$ is the Bloch vector which is a null vector and the correlation matrix is $T^b=
(c_1,c_2,c_3)$
with $-1\leqslant \text{c}_{i} \leqslant1$.\\

The eigenvalues  $\lambda_{1}$, $\lambda_{2}$, $\lambda_{3}$, $\lambda_{4}$  of Bell diagonal states  are expressed as, $\lambda_{1}=\frac{1}{4}(\chi - 2c_{1})$, $\lambda_{2}=\frac{1}{4}(\chi - 2c_{2})$, $\lambda_{3}=\frac{1}{4}(\chi - 2c_{3})$, $\lambda_{4}=\frac{1}{4}(2 - \chi)$, where $\chi$ = 1 +  $c_{1}$ + $c_{2}$ + $c_{3}$. Therefore, necessary and sufficient condition for a Bell diagonal state to lie in \textbf{ACVENN} is given by  $S(\sigma_{bell}) \geq 1$ which in terms of $c_{1}$, $c_{2}$, $c_{3}$ and $\chi$ becomes
\begin{equation}
\begin{split}
    \label{BDstateBlochparameterscondition}
    \log((\chi - 2c_{2})(\chi - 2c_{3})(2 - \chi)(\chi - 2c_{1})) \\
   +  c_{1}\log(\frac{(\chi - 2c_{2})(\chi - 2c_{3})}{(2 - \chi)(\chi - 2c_{1})}) \\
   +  c_{2}\log(\frac{(\chi - 2c_{3})(\chi - 2c_{1})}{(\chi - 2c_{2})(2 - \chi)}) \\
   +  c_{3}\log(\frac{(\chi - 2c_{2})(\chi - 2c_{1})}{(\chi - 2c_{3})(2 - \chi)}) \leq 4.
\end{split}
\end{equation}

\begin{figure}[h]
    \centering
    \includegraphics[width=70mm,scale=0.5]{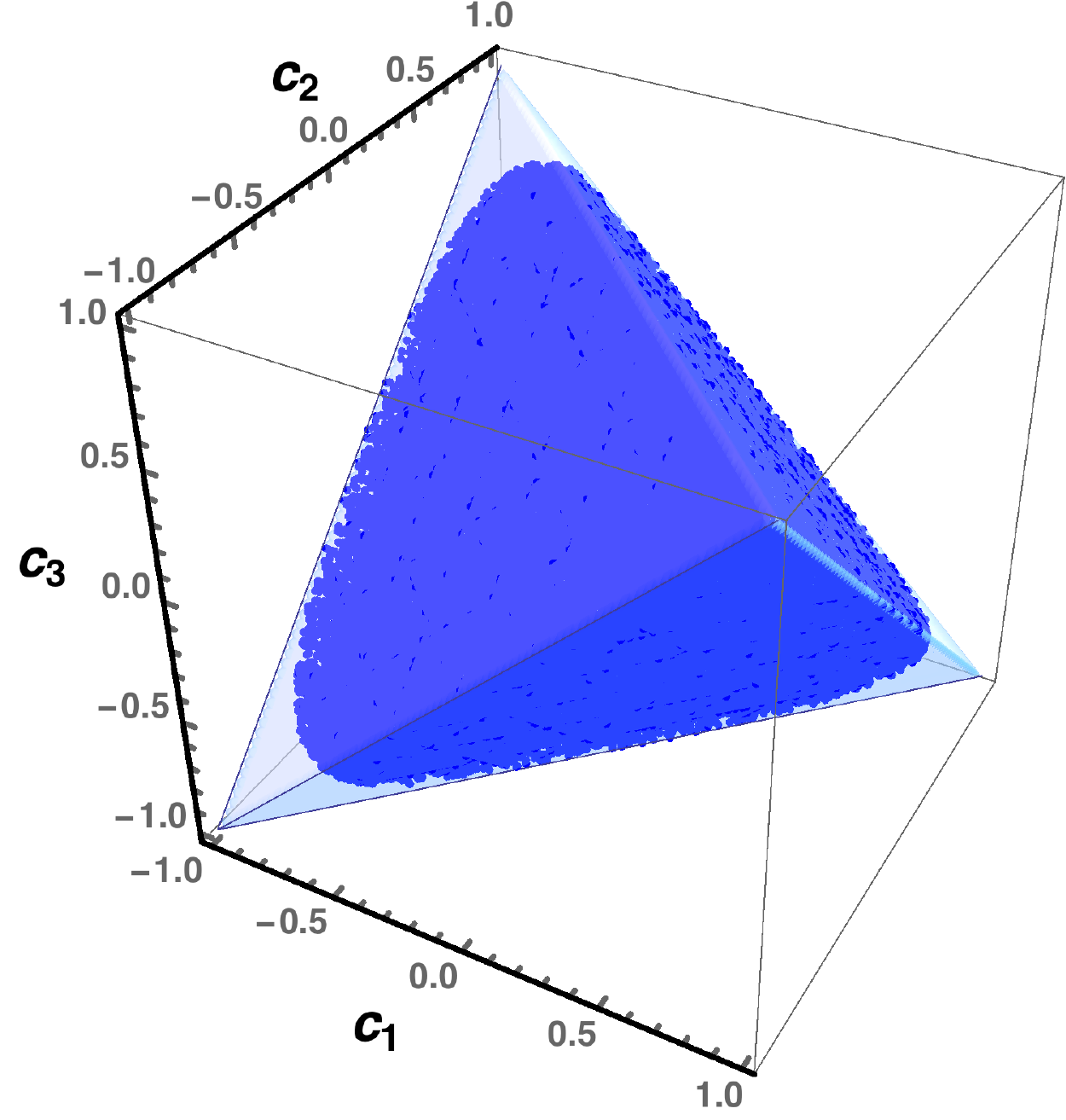}
    \caption{The von Neumann entropy of Bell diagonal state $\sigma_{bell}$ against the parameter $c_i$}
    \label{fig:BDtetra1}
\end{figure}

In figure \ref{fig:BDtetra1}, we consider an exhaustive ensemble  of  $10^5$ states within which the dark blue colour area at the centre of the octahedron determines the class of states for which  $S(\sigma_{bell}) \geq 1$ and falls into our \textbf{ACVENN} class. The light blue areas at the corner are those areas whose conditional entropy can be made negative after the application of some global unitary transformation. It is evident from figure \ref{fig:BDtetra1} that the non negativity of conditional entropy for most of the part of the Bell diagonal states remains invariant after the application of global unitary transformation.
\newline



\section{\label{sec:CONVEXandCOMP} Convexity and Compactness of the ACVENN class: Existence of Witness}

\noindent In this section we show that the \textbf{ACVENN} class which is a subset of the class \textbf{Q} is a convex and compact set. This helps us identifying the states whose conditional entropy remains negative even after the application of global unitary. We now present the proof that the set  \textbf{ACVENN} is convex and compact.\\

\begin{figure}
    \centering
    \includegraphics[width=70mm,scale=0.5]{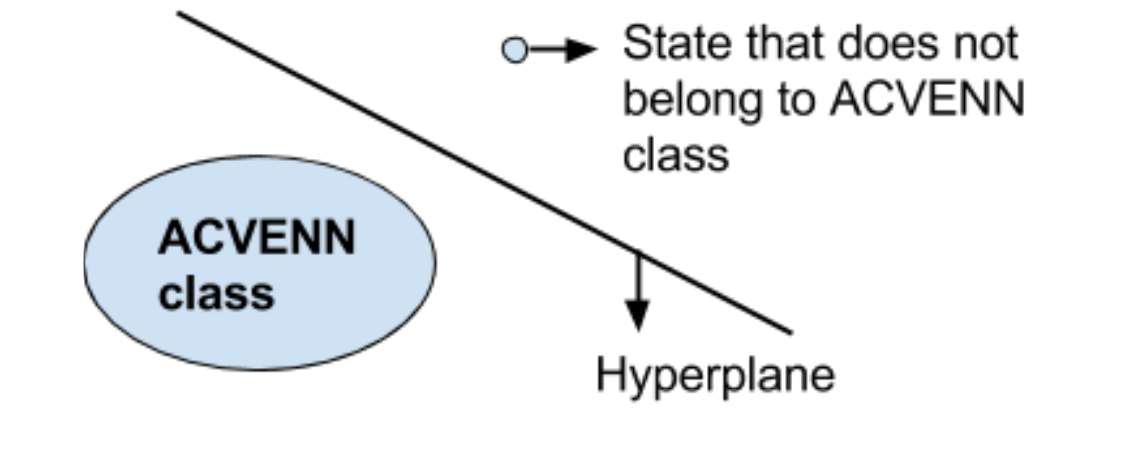}
    \caption{The set ACVENN is convex and compact,
and using the Hahn-Banach theorem \cite{Hahn} it follows that any state
not belonging to ACVENN can be separated from the states that belong to ACVENN by a hyperplane, thus providing for
the existence of a witness.}
    \label{fig:Witness_label}
\end{figure}

\noindent \textbf{Existence of Witness:} The theorems below will support the existence of witness operators to classify \textbf{ACVENN} states from states which are not in \textbf{ACVENN}. See figure \ref{fig:Witness_label} for reference.

\begin{theorem}
\textbf{ACVENN} \textit{is convex.} 
\end{theorem}
\begin{proof}
Consider $ \sigma_1,\sigma_2 \in \textbf{ACVENN}$. Therefore, $ S(\sigma_i) \ge 1 , i=1,2$ . Now by the concavity of von Neumann entropy $ S(\lambda \sigma_1 + (1-\lambda) \sigma_2) \ge 1 $, where $ \lambda \in [0,1] $. Hence, $ \lambda \sigma_1 + (1-\lambda) \sigma_2 \in \textbf{ACVENN}  $ , implying \textbf{ACVENN} is convex.
\end{proof}
\begin{theorem}
\textbf{ACVENN} \textit{is compact subset of} \textbf{Q}.
\end{theorem}
\begin{proof}
Let us define a function \textit{f} : \textbf{Q} $\rightarrow$ $\mathbb{R}$ as
\begin{equation}
\label{function1}
    f(\rho)=S(\rho),
\end{equation}
as \textbf{ACVENN} = \{ $\sigma_{ac}$ : S($\sigma_{ac}$)$\geq$1 \}, and $f$ will have a maximum value of 2, we can say \textbf{ACVENN} = $f^{-1}$[1,2]. $f$ is a continuous function as $S$ is a continuous function \cite{fannes}. Therefore, \textbf{ACVENN} = $f^{-1}$[1,2] is a closed set in \textbf{Q} defined under the trace norm. The set \textbf{ACVENN} is bounded as every density matrix has a bounded spectrum, i.e., their eigen values lies between 0 and 1. This proves that the \textbf{ACVENN} class is compact.
\end{proof}
\noindent The theorem now guarantees the existence of Hermitian operators to succesfully identify states that do not belong to \textbf{ACVENN}.\\

\noindent Next we estimate the size of the \textbf{ACVENN} class by taking the maximum and minimum distance from the identity ($\frac{\mathbb{I}}{2}\otimes\frac{\mathbb{I}}{2}$). The distance measure we have used in this context is the Frobenius norm which is given by $\|X\|$ = $\sqrt{Tr(X^\dagger X)}$. Having already proved that \textbf{ACVENN} is a convex set, we try to find out the maximum and minimum distance from $\frac{\mathbb{I}}{2}\otimes\frac{\mathbb{I}}{2}$.
\newline
For any general $\widetilde{\varrho}$, distance from $\frac{\mathbb{I}}{2}\otimes\frac{\mathbb{I}}{2}$ is given by $\|\widetilde{\varrho} - \frac{\mathbb{I}}{4}\| = \sqrt{Tr((\widetilde{\varrho} - \frac{\mathbb{I}}{4})^\dagger (\widetilde{\varrho} - \frac{\mathbb{I}}{4}))}$, which on solving further results to $\sqrt{Tr(\widetilde{\varrho}^2) - \frac{1}{4}}$.
\newline
To calculate the maximum distance we needed to maximise $|\sigma - \frac{\mathbb{I}}{4}\|$, over all $\sigma$ $\in$ \textbf{ACVENN}. Here we  solve this problem numerically. After going through $2 \times 10^5$  \textbf{ACVENN} states the maximum distance we have is 0.645966 by the state whose eigen values were $\lambda_{1}$ = 0.809161, $\lambda_{2}$ = 0.0521141, $\lambda_{3}$ = 0.0595448, $\lambda_{4}$ = 0.0791805.
\newline
To calculate the minimum distance we needed to minimise $|\rho - \frac{\mathbb{I}}{4}\|$, over all $\rho$ $\notin$ \textbf{ACVENN}. Going through $1 \times 10^5$ \textbf{non-ACVENN} states numerically, we attained the minimum distance as 0.507225. This is given by a state whose eigen values were $\lambda_{1}$ = 0.00014347, $\lambda_{2}$ = 0.000551157, $\lambda_{3}$ = 0.436523, $\lambda_{4}$ = 0.562783.
\newline
In figure \ref{fig:DistanceACVENN} we show rough estimation of the size of \textbf{ACVENN} class.
\begin{figure}[h]
    \centering
    \includegraphics[width=70mm,scale=0.5]{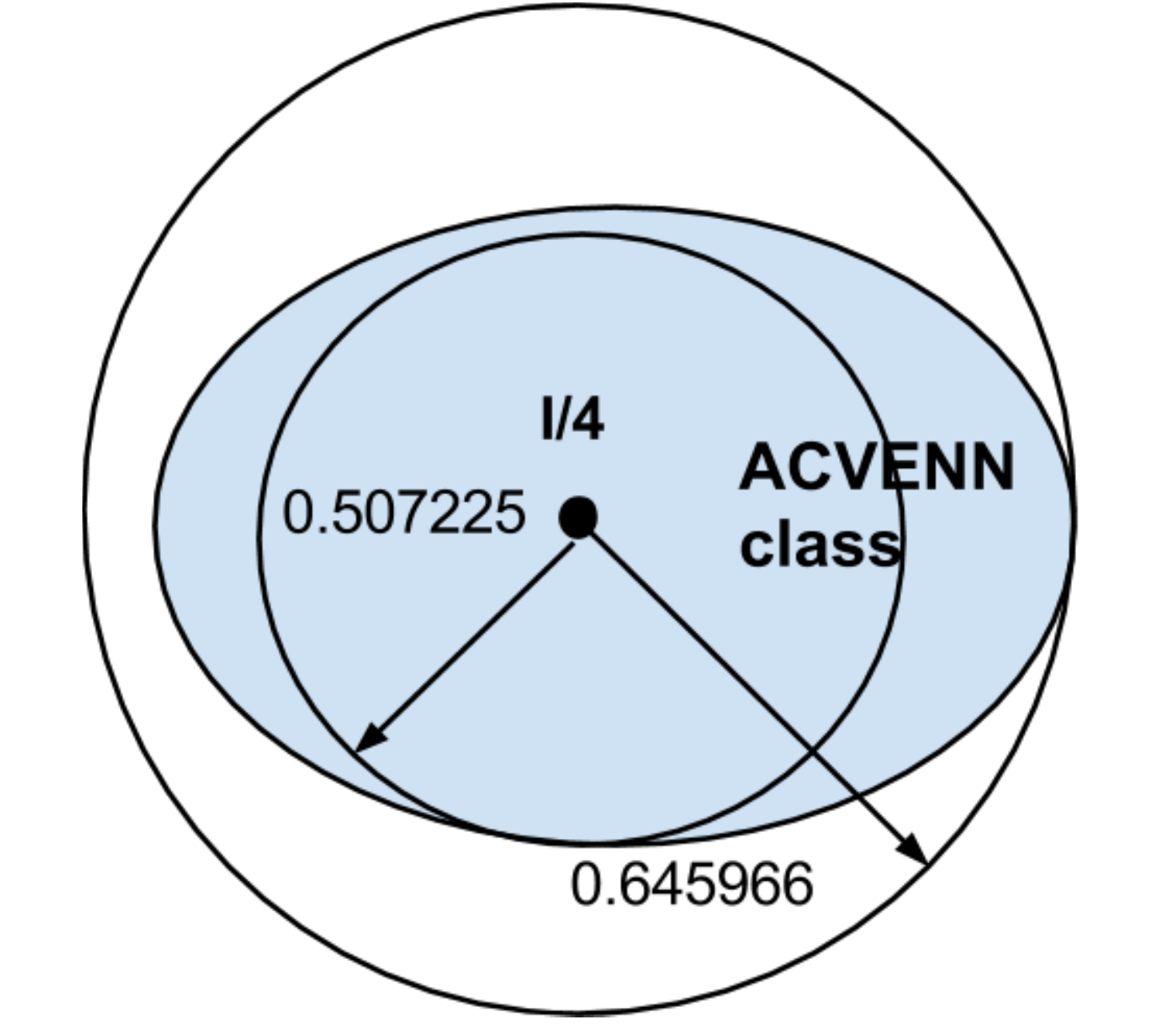}
    \caption{The figure depicts the approximate size of the \textbf{ACVENN} class}
    \label{fig:DistanceACVENN}
\end{figure}

\section{\label{sec:REL}Relation between AS, ACVENN and AL}

\noindent In this section we give a comparative picture of three classes of states. These classes remain invariant from the context of separability \textbf{AS}, non violation of Bell's inequlity \textbf{AL} and the non-negative conditional entropy \textbf{ACVENN} under global unitary transformation.

\subsection{\label{sec:REL_ASvsACVENN}AS vs ACVENN }
The figure \ref{fig:ACVENNvsASvsS} shows the relation between \textbf{AS} , \textbf{ACVENN} and separable states.

\begin{figure}[h]
    \centering
    \includegraphics[width=70mm,scale=0.5]{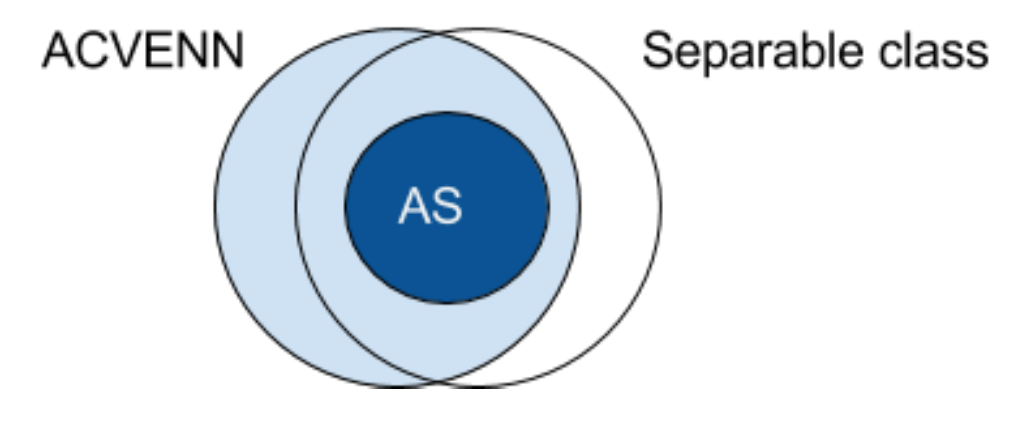}
    \caption{The figure depicts the relation between \textbf{ACVENN}, \textbf{AS} and separable class of states}
    \label{fig:ACVENNvsASvsS}
\end{figure}

\begin{lemma}
\label{lem:ASvsACVENN}
 \textbf{AS} $\subseteq$ \textbf{ACVENN}
\end{lemma}

\begin{proof}
Absolutely separable states preserve separability under any global unitary action. The non-negativity of conditional entropies is a necessary condition for separability \cite{nonNegativeCVN}. All separable states have a non-negative conditional von Neumann entropy. Absolutely separable states(\textbf{AS}) remain separable under any unitary transformation. \textbf{AS} will always have non-negative conditional von Neumann entropy. So it will be a subset of \textbf{ACVENN} class.
\end{proof}

\subsubsection{\label{sec:REL_ASvsACVENN_ILLUS1}Illustration: A. \noindent Absolutely separable Werner states}
Let us consider the Werner states $\sigma_{wer}$ = p$| \psi \rangle \langle \psi |$ + $\frac{1-p}{4}\mathbb{I}$, where, $| \psi \rangle $ is the Bell state $\frac{1}{\sqrt{2}}(| 00 \rangle + | 11 \rangle$). The state $\sigma_{wer}$ belongs to \textbf{ACVENN} for values of $p$ that satisfies the equation, $3(1-p)\log(1-p) + (1 + 3p)\log(1 + 3p) \leq 4$. Solving the inequality we get $p$ $\in$ [0,$\approx$0.7476] as obtained earlier. 
For the states to be in \textbf{AS}, one must have $a_1$ $\leq$ $a_3$ + 2$\sqrt{a_2a_4}$, where $a_1$=$\frac{(1+3p)}{4}$, $a_2$=$\frac{(1-p)}{4}$, $a_3$=$\frac{(1-p)}{4}$, $a_4$=$\frac{(1-p)}{4}$ are all the eigen values in descending order. $\sigma_{wer}$ belongs to \textbf{AS} for values of p $\in$ [0,$\frac{1}{3}$]. This gives an example of an absolutely separable state \textbf{AS} which is contained in \textbf{ACVENN} class. \\

\subsubsection{\label{sec:REL_ASvsACVENN_ILLUS2} Illustration: B. Incoherent in Computational Basis} 
Next, we give an example of a class of states which are incoherent in the computational basis,
\begin{equation}
\label{compu_diag}
\sigma_{comp}=a_1 | 00 \rangle \langle 00 + a_2 | 01 \rangle \langle 01 | + a_3 | 10 \rangle \langle 10 | + a_4 | 11 \rangle \langle 11 |.
\end{equation}

Taking the example of a state with eigenvalues $a_1$ = $\frac{5}{10}$, $a_2$ = $\frac{3}{10}$, $a_3$ = $\frac{2}{10}$, $a_4$ = 0, S($\sigma_{comp}$) = -$\sum$ $a_i$$\log_{2}$$a_i$ $\approx$ 1.485 $\geq$ 1. Hence, this state $\in$ \textbf{ACVENN}.
\newline
We see that $a_1$ $\geq$ $a_2$ $\geq$ $a_3$ $\geq$ $a_4$. For the states to be in \textbf{AS}, one must have $a_1$ $\leq$ $a_3$ + 2$\sqrt{a_2a_4}$. In this specific case $a_1$ = 0.5, $a_3$ + 2$\sqrt{a_2a_4}$ = 0.2. Clearly, this shows that \textbf{AS} is a subset of \textbf{ACVENN} .

\begin{theorem}
\textbf{AS} $\subset$ \textbf{ACVENN}
\end{theorem}
\begin{proof}
In lemma \ref{lem:ASvsACVENN} it has been shown that \textbf{AS} $\subseteq$ \textbf{ACVENN}. In fact, we can say more than that. In view of the example on Werner states in Illustration \ref{sec:REL_ASvsACVENN_ILLUS1}, we have seen that there are states that donot belong to \textbf{AS} but belong to \textbf{ACVENN}. This shows that absolutely separable states (\textbf{AS}) form a proper subset of \textbf{ACVENN} 
\end{proof}

\noindent After proving that \textbf{AS} $\subset$ \textbf{ACVENN} we want to estimate the minimum and maximum entropy recorded by the states belonging to \textbf{AS}. We solve this problem numerically as well. After going through $1 \times 10^5$  \textbf{AS} states the minimum entropy that we obtain is 1.58662. This is attained for a state with eigen values $\lambda_{1}$ = 0.341023, $\lambda_{2}$ = 0.331417, $\lambda_{3}$ = 0.327411, $\lambda_{4}$ = 0.000148614. We already know that the maximum entropy for \textbf{AS} is 2. 
This gives us a rough estimate of the volume of \textbf{AS} states lying within the \textbf{ACVENN} in terms of entropy.

\subsection{\label{sec:REL_ALvsACVENN} AL vs ACVENN }
The Werner states are absolutely local for the visibility factor $ p \le 1/\sqrt{2} $, and they belong to \textbf{ACVENN} for $ p \le 0.7476 $. Therefore, the \textit{absolutely Bell-CHSH local} Werner states form a subset of the \textbf{ACVENN} class. This is an interesting result as that would mean, there are states that violate \textit{Bell-CHSH} inequality, and still under any unitary cannot be improved to a state with negative conditional entropy.\\
However, it is difficult to comment in general on the relation between \textbf{AL} and \textbf{ACVENN} class.

\section{\label{sec:APPL}APPLICATIONS: STATE MERGING AND SUPER DENSE CODING}

\noindent In this section we show how characterizing this \textbf{ACVENN} class of states helps to identify the states which are not useful for information processing tasks like super dense coding and state merging are made useful with the help of global unitary transformations. In either of these tasks we are able to detect a class of states which can be converted into super dense coding and state merging resource by applying global unitary transformations. In figure \ref{fig:CompreWitness} we give pictorial description of two types of witness operators that can be created. $W_{S_{D}}$ and $W_{S_{M}}$ act as hyperplanes that detect states useful for superdense coding and state merging respectively from the states belonging to \textbf{ACVENN} class which can never be made useful for these information theoretic tasks by using global unitary operations. 

\begin{figure}[h]
    \centering
    \includegraphics[width=70mm,scale=0.5]{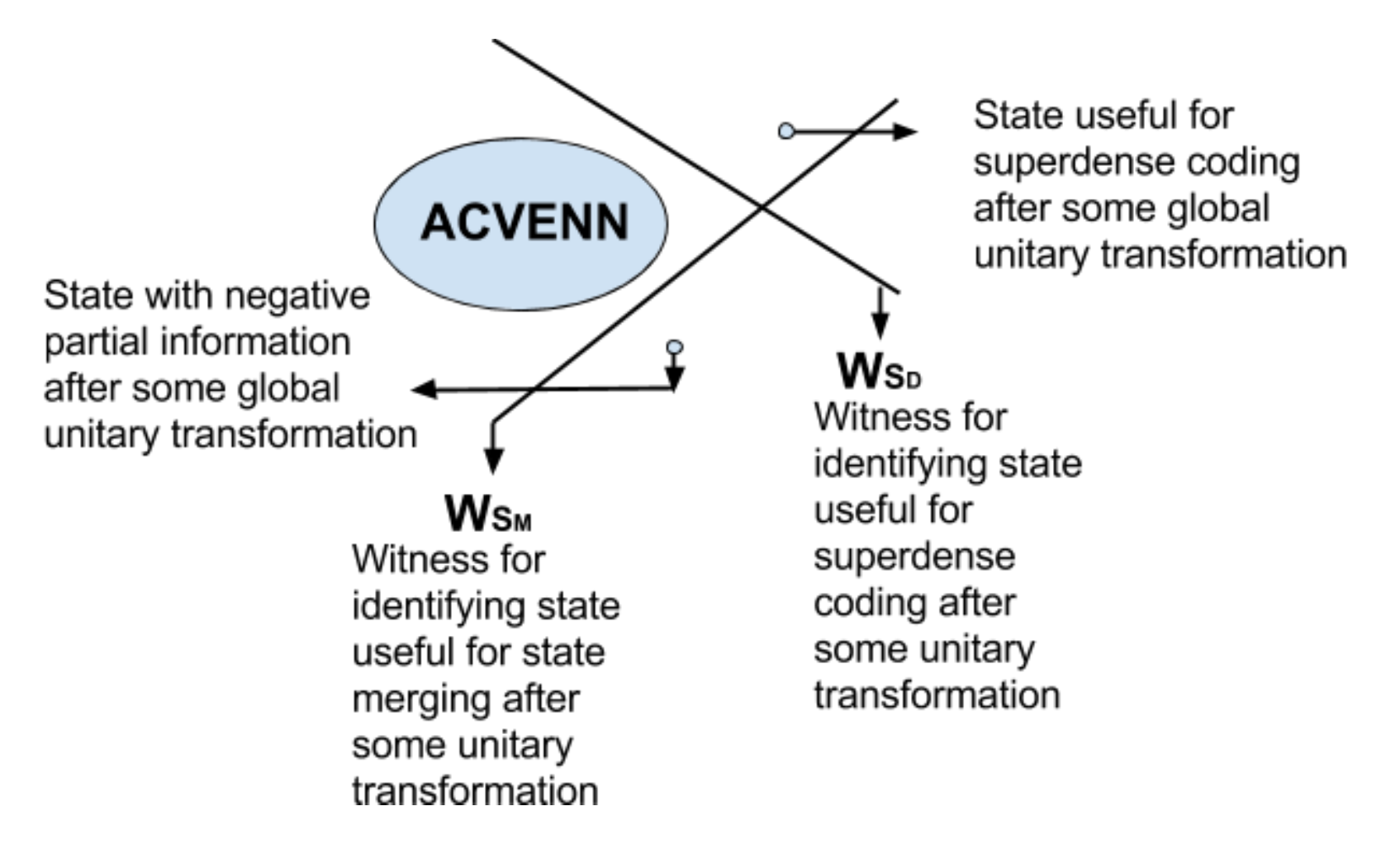}
    \caption{The role of witness in detecting the states which are useful for super dense coding and state merging after global unitary transformation}
    \label{fig:CompreWitness}
\end{figure}

\subsection{\label{sec:APPL_SDC}Super Dense coding}

\noindent In particular, super dense coding capacity for a mixed state $\rho_{AB}$ in $D(H_{d}\otimes H_{d})$ is defined by
\begin{equation}
 {\cal C}_{AB}= max\{\log_2d ,\\  \log_2d+S(\rho_B)-S(\rho_{AB})\},
\end{equation}
where, $\rho_B=tr_A[\rho_{AB}]$ \cite{ben2,distdc}. ${\cal C}_{AB}$ is nothing but the amount of classical information that can be sent from system $A$ to system $B$. Here we note that the expression  $S(\rho_{B})-S(\rho_{AB})$ can either be 
positive or negative. If it is positive then one can use the shared state to transfer bits greater than the classical limit of 
$\log_2d$ bits. This in particular known as the quantum advantage where we can do more than the classical limit. For pure states, $S(\rho_{AB})=0$, then the 
super dense coding capacity is given by,
\begin{equation}
 {\cal C}_{AB}=\log_2d+S(\rho_B)=\log_2d+E(\rho_{AB}),
\end{equation}
where, the entanglement entropy $E(\rho_{AB})$ of a pure state $\rho_{AB}$ is nothing but the von Neumann entropy $S(\rho_{B})$ of the reduced subsystem $\rho_B$. The capacity will be maximum for the Bell states as $S(\rho_{B})$ will be equal to $1$. In a nutshell a state $\rho_{AB}$ for which this expression $S(\rho_{B})-S(\rho_{AB})$ is positive will give us a quantum advantage for superdense coding. In other words, a state with a negative conditional entropy $S(A|B)$ will be useful. It is obvious that not all states will have negative conditional entropy. The next important question is if we apply global unitary operator can we make a state which is not useful for super dense coding to a useful resource. In other words whether we can change the conditional entropy of the state from positive to negative. The answer is yes, however there will be some states for which we can not do that. These set of invariant states are nothing but previously described \textbf{ACVENN} class of states which can never be useful from the perspective of super dense coding, even after the application of global unitary operators. As we have seen previously that this class of state is convex and compact, then in principle it will be possible to create a witness operator ($W_{S_{D}}$ as seen in figure \ref{fig:CompreWitness}) to detect the states which are initially not useful but made useful for superdense coding. It is important to mention here that this witness operator is not an witness operator to detect the states which are useful for super dense coding as opposed to the non useful state. Class of states useful for super dense coding is not a convex and compact set, so is the class of states that are not useful for super dense coding. This witness operator detects those states which need not be useful initially but can be made useful after global unitary transformation.  Further we provide example to show all these kind of states.

\subsubsection{\label{sec:APPL_SDC_ILLUS}Illustrations}
For our first example let's consider a mixed separable state in $D(H_{2}\otimes H_{2})$ given by \cite{wit}
\begin{eqnarray}
\label{eq:chi_mix}
\rho=\begin{vmatrix}a & 0 & b & 0 \\
  0 & 0 & 0 & 0 \\
  b & 0 & 1-a & 0 \\
  0 & 0 & 0 & 0 
\end{vmatrix}.
\end{eqnarray}
The eigen values for this state are $\frac{1-q}{2}$, $\frac{1+q}{2}$, $0$, $0$ and eigen values of the subsystem A are $\frac{1-q}{2}$, $\frac{1+q}{2}$ where $q=\sqrt{1-4a+4a^2+4b^2}$. The state $\rho$ $\in$ \textbf{ACVENN} iff $S(\rho)$ $\geq$ 1. In the current scenario that occurs only when $q=0$. However, for no real values $a$ and $b$ is $q=0$. Therefore we know that for no real values of $a$ and $b$ does this state belong to \textbf{ACVENN}. Thus $S(\rho_{A|B})$=0 for all real values of $a$ and $b$, which is clearly not having negative conditional entropy, therefore, providing no quantum advantage.
But, this on application of the unitary operator,
\begin{eqnarray}
\label{eq:U_chi}
U_{1}=\frac{1}{\sqrt{2}}\begin{vmatrix}1 & 0 & 0 & 1 \\
  0 & \sqrt{2} & 0 & 0 \\
  0 & 0 & \sqrt{2} & 0 \\
  -1 & 0 & 0 & 1
\end{vmatrix},
\end{eqnarray}
becomes,
\begin{eqnarray}
\label{eq:chi_mix}
\rho^{'}=\begin{vmatrix}\frac{a}{2} & 0 & \frac{b}{\sqrt{2}} & \frac{-a}{2} \\
  0 & 0 & 0 & 0 \\
  \frac{b}{\sqrt{2}} & 0 & 1-a & \frac{-b}{\sqrt{2}} \\
  \frac{-a}{2} & 0 & \frac{-b}{\sqrt{2}} & \frac{a}{2} 
\end{vmatrix}.
\end{eqnarray}
While the eigen values of $\rho^{'}$ remains unchanged, the eigen values of the subsystem $\rho^{'}_{A}$ becomes $\frac{1-q^{'}}{2}$, $\frac{1+q^{'}}{2}$ where $q^{'}=\sqrt{1-2a+a^2+2b^2}$. For all the values of $a$ and $b$, where $q$ $>$ $q^{'}$ the state can be made useful for superdense coding. One such example would be when $a=0.5$ and $b=0.4$. Thus, we provide an example of a state which was not useful for super dense coding initially but after a unitary transformation it was made useful for super dense coding. \\

\subsection{\label{sec:APPL_SMER}State merging and Partial Quantum Information}

\noindent In the information processing scenarios, it is important to ask this question: if an unknown  quantum state distributed over two systems say $A$ and $B$, how much quantum communication is needed to transfer the full state to one system. This communication measures the \textit{partial information} one system needs conditioned on it’s prior information. Remarkably, this is given by the conditional entropy $S(A|B)$ (if it is from $A$ to $B$) of the system. It is interesting to note that in principle this entropy can be positive ($S(A|B)>0$), negative ($S(A|B)<0$) and zero ($S(A|B)=0$), where each have different meaning in the context of state merging. If the partial information is positive, its sender needs to communicate this number of quantum bits to the receiver; if zero it tells there is no need of such communication; if it is negative, the sender and receiver instead gain the corresponding potential for future quantum communication. So given a quantum state $\rho_{AB}$, shared between $A$ and $B$, three possible cases arise, and we characterize the state based on these cases, namely for states with $S(A|B)>0$ we denote them $\rho_{S(A|B)>0}$ and similarly other states as $\rho_{S(A|B)<0}$, $\rho_{S(A|B)=0}$.  It is always useful from the information theoretic point of view to look out for the states $\rho_{S(A|B)<0}$ as they have the potential for future communication. It is needless to mention that not all states will be of this type. So the next question that becomes important in this context of global unitary is to find out states which are initially of the type  $\rho_{S(A|B)>0}$ but can be converted to the type $\rho_{S(A|B)<0}$ after global unitary operations. Those states for which the conditional entropy remains positive even after all possible global unitary operations is nothing but the previously defined \textbf{ACVENN} class. Since we have already proved that \textbf{ACVENN} class is always convex and compact, this means that we can detect the states whose partial information can be made negative after the global unitary operation with the help of witness operator ($W_{S_{M}}$ as seen in figure \ref{fig:CompreWitness}). Like in the case of super dense coding  it is important to mention here also, that we are not showing that the set $\rho_{S(A|B)>0}$ is convex and compact but the set for which partial information remains positive after the global unitary operation (say, $U\rho_{S(A|B)}U^{\dagger}>0$ )  is convex and compact. As a result we are not detecting the state for which the partial information is negative instead those states whose partial information can be made negative (say, $U\rho_{S(A|B)}U^{\dagger}<0$ ) after the application of global unitary. We give examples to identify all these classes.

\subsubsection{\label{sec:APPL_SMER_ILLUS}Illustrations}

Lets take a mixed two qubit state,
\begin{equation}
\label{chi_ya}
\rho_{AB}=\frac{3}{4} | 00 \rangle \langle 00 | + \frac{1}{4} |11 \rangle \langle 11 |,
\end{equation}
we see, $S(\rho_{B|A})$=0. After the application of a unitary transformation $U_{2}$=$U^{-1}_{1}$, where $U_{1}$ is defined in equation \ref{eq:U_chi} above. The state $\rho_{AB}$ transforms to,
\begin{equation}
\label{chii_ya}
\rho^{'}_{AB}=\frac{1}{2} | 00 \rangle \langle 00 | + \frac{1}{4} |00 \rangle \langle 11 | +\frac{1}{4} | 11 \rangle \langle 00 | + \frac{1}{2} | 11 \rangle \langle 11 |.
\end{equation}
The $S(\rho^{'}_{B|A})$ is $-0.1887$. Thus we give an example of a state which has non negative conditional entropy initially can be made negative with the help of a unitary transformation  and also in principle one can construct the witness operator to detect such kind of states.\\

In this subsection we also ask this question: \textit{For a given spectrum of density matrix for which states the minimum state merging cost will be achieved?} \\

\begin{theorem}
For a given spectrum of density matrix the minimum state merging cost will be achieved at the Bell diagonal  states .
\end{theorem}

\begin{proof}
We know that the conditional entropy for the quantum state $\rho_{AB}$ is given by the difference, $S(B|A)=S(\rho_{AB})- S(\rho_A)$. Let us assume that the spectrum $\rho_{AB}$ is fixed with eigen values $a_i$ , $i=1,2,3,4$. Since the spectrum is fixed we have freedom to apply the global unitary operator. Now the question is to minimize $S(B|A)$, by using only global unitary operations. Since the global unitary operations will not change $S(\rho_{AB})$, we need to maximize $S(\rho_A)$. Now $S(\rho_A)$ is maximized if $\rho_A=\mathbb{I}/2$. The reduced density matrices for Bell diagonal state is $\mathbb{I}/2$. Therefore if one reaches Bell diagonal state by some global unitary no further maximization of $S(\rho_A)$ is possible. Thus for a given spectrum of density matrices the minimal merging cost is attained at the Bell-diagonal state.
\end{proof}

\section{\label{sec:CONCL}CONCLUSION}

\noindent In this work for a general two qubit system we are able to charecterize class of states \textbf{ACVENN} whose von Neumann entropy will remain positive even after the application of global unitary operator. More specifically, we are able to show that the staes with von Neumann entropy greater than $1$, is the same \textbf{ACVENN} class of states.

We also  find that this class of states is convex and compact, which guarantees the existence of witness for detecting the states which could have initially positive conditional entropy but have negative conditional entropy after the application of unitary operator. This in turn gives the power to identify the states which are not initially useful but can be made useful in the information processing tasks like superdense coding and state merging.

\section{\label{sec:ACKN}ACKNOWLEDGEMENT}

We gratefully acknowledge Prof. Guruprasad Kar for his enriching suggestions.

\end{document}